\documentclass{LMCS}

\def\doi{9(1:08)2013}
\lmcsheading%
{\doi}
{1--22}
{}
{}
{Aug.~12, 2012}
{Mar.~\phantom04, 2013}
{}

\usepackage{enumerate}
\usepackage{hyperref}

\usepackage{amsmath}
\usepackage{amssymb}
\usepackage{graphicx}
\usepackage{tikz}
\usetikzlibrary{calc,automata}
\usepackage{wrapfig}

\theoremstyle{plain}\newtheorem{theorem}[thm]{Theorem}
\theoremstyle{plain}\newtheorem{proposition}[thm]{Proposition}
\theoremstyle{plain}\newtheorem{lemma}[thm]{Lemma}
\theoremstyle{plain}\newtheorem{corollary}[thm]{Corollary}
\theoremstyle{definition}\newtheorem{example}[thm]{Example}

\begin{document}

\title[On the Complexity of Equivalence and Minimisation for
  $\mathbb{Q}$-Weighted Automata] {On the Complexity of Equivalence
  and Minimisation for $\mathbb{Q}$-Weighted Automata\rsuper*}

\author[S. Kiefer]{Stefan Kiefer\rsuper a}
\address{{\lsuper{a,c,d,e}}University of Oxford, UK}
\email{\{stekie, joel, Bjoern.Wachter, jbw\}@cs.ox.ac.uk}

\author[A.S.~Murawski]{Andrzej S.~Murawski\rsuper b} 
\address{{\lsuper b}Department of Computer Science, University 
of Warwick, UK}
\email{A.Murawski@warwick.ac.uk}

\author[J.~Ouaknine]{Jo\"el Ouaknine\rsuper c}   
\address{\vskip-6 pt}

\author[B. Wachter]{Bj\"orn Wachter\rsuper d} 
\address{\vskip-6 pt}

\author[J.~Worrell]{James Worrell\rsuper e}   
\address{\vskip-6 pt}



\ACMCCS{[{\bf Theory of computation}]: Design and analysis of
  algorithms---Approximation algorithms analysis---Numeric
  approximation algorithms \& Semantics and reasoning---Program
  reasoning---Program verification}

\titlecomment{{\lsuper*}This is a full and improved version of the FoSSaCS'12
  paper with the same title.  An algorithm from the same authors'
  CAV'11 paper~\cite{Cav11} was incorporated in
  Section~\ref{sub:cav-algorithm} and new algorithms for minimisation were added in Section~\ref{sub:minimisation}.
  Section~\ref{sub:exp-equiv} is also new.}



\renewcommand{\P}[1]{{\cal P}\left(#1\right)}
\newcommand{\Q}{\mathbb{Q}}
\newcommand{\F}{\mathbb{F}}
\newcommand{\R}{\mathbb{R}}
\newcommand{\N}{\mathbb{N}}
\newcommand{\Z}{\mathbb{Z}}
\newcommand{\norm}[1]{\lVert#1\rVert}
\newcommand{\A}{\mathcal{A}}
\newcommand{\B}{\mathcal{B}}
\newcommand{\C}{\mathcal{C}}
\newcommand{\V}[2]{{}_{#1}\!V_{#2}}
\newcommand{\ns}[1]{n^{(#1)}}
\newcommand{\Ms}[1]{M^{(#1)}}
\newcommand{\ind}{\mbox{}\hspace{7mm}}
\newcommand{\indd}{\ind\ind}
\newcommand{\inddd}{\indd\ind}
\newcommand{\indddd}{\inddd\ind}
\newcommand\apex[1]{\textsc{apex}}
\newcommand{\va}{\boldsymbol{a}}
\newcommand{\valpha}{\boldsymbol{\alpha}}
\newcommand{\alphas}[1]{\valpha^{(#1)}}
\newcommand{\Gammas}[1]{\Gamma^{(#1)}}
\newcommand{\AF}{\overrightarrow{\mathcal{A}}}
\newcommand{\AB}{\overleftarrow{\mathcal{A}}}
\newcommand{\AFB}{\overleftarrow{\overrightarrow{\mathcal{A}}}}
\newcommand{\ABF}{\overrightarrow{\overleftarrow{\mathcal{A}}}}
\newcommand{\FBF}{\overrightarrow{\overleftarrow{F}}}

\newcommand{\BV}{\mathsf{B}}
\newcommand{\Bh}{\widehat{B}}
\newcommand{\etaBF}{\overleftarrow{\overrightarrow{\eta}}}
\newcommand{\FV}{\mathsf{F}}
\newcommand{\Ft}{\widetilde{F}}
\newcommand{\Fh}{\widehat{F}}
\newcommand{\FF}{\overrightarrow{F}}
\newcommand{\BB}{\overleftarrow{B}}
\newcommand{\nF}{\overrightarrow{n}}
\newcommand{\nB}{\overleftarrow{n}}
\newcommand{\nBF}{\overrightarrow{\overleftarrow{n}}}
\newcommand{\MF}{\overrightarrow{M}}
\newcommand{\MB}{\overleftarrow{M}}
\newcommand{\MBF}{\overrightarrow{\overleftarrow{M}}}

\newcommand{\vb}{\boldsymbol{b}}
\newcommand{\bs}[1]{\vb^{(#1)}}
\newcommand{\vc}{\boldsymbol{c}}
\newcommand{\ve}{\boldsymbol{e}}
\newcommand{\es}[1]{\ve^{(#1)}}
\newcommand{\veta}{\boldsymbol{\eta}}
\newcommand{\etas}[1]{\veta^{(#1)}}
\newcommand{\vh}{\boldsymbol{h}}
\newcommand{\hs}[1]{\vh^{(#1)}}
\newcommand{\vl}{\boldsymbol{\ell}}
\newcommand{\ls}[1]{\ell^{(#1)}}
\newcommand{\ps}[1]{p^{(#1)}}
\newcommand{\vu}{\boldsymbol{u}}
\newcommand{\us}[1]{\vu^{(#1)}}
\newcommand{\vv}{\boldsymbol{v}}
\newcommand{\vs}[1]{\vv^{(#1)}}
\newcommand{\vr}{\boldsymbol{r}}
\newcommand{\vx}{\boldsymbol{x}}
\newcommand{\vz}{\boldsymbol{z}}
\newcommand{\vw}{\boldsymbol{w}}
\newcommand{\vzero}{\boldsymbol{0}}
\newcommand{\parikh}[1]{#1{\downarrow}}
\newcommand{\Span}{\mathrm{span}}
\newcommand{\rank}{\ensuremath{\textup{rank}}}


\begin{abstract}
\noindent
This paper is concerned with the computational complexity of
equivalence and minimisation for automata with transition weights in
the ring $\mathbb{Q}$ of rational numbers.  We use polynomial identity
testing and the Isolation Lemma to obtain complexity bounds, focussing
on the class \textbf{NC} of problems within \textbf{P} solvable in
polylogarithmic parallel time.  For finite $\mathbb{Q}$-weighted
automata, we give a randomised \textbf{NC} procedure that either
outputs that two automata are equivalent or returns a word on which
they differ.  We also give an \textbf{NC} procedure for deciding
whether a given automaton is minimal, as well as a randomised
\textbf{NC} procedure that minimises an automaton.  We consider
probabilistic automata with rewards, similar to Markov Decision
Processes.  For these automata we consider two notions of equivalence:
expectation equivalence and distribution equivalence.  The former
requires that two automata have the same expected reward on each input
word, while the latter requires that each input word induce the same
distribution on rewards in each automaton.  For both notions we give
algorithms for deciding equivalence by reduction to equivalence of
$\mathbb{Q}$-weighted automata.  Finally we show that the equivalence
problem for $\mathbb{Q}$-weighted visibly pushdown automata is
logspace equivalent to the polynomial identity testing problem.
\end{abstract}

\maketitle

\section{Introduction}
Probabilistic and weighted automata were introduced in the 1960s, with
many fundamental results established in the papers of
Schutzenberger~\cite{Schutzenberger} and Rabin~\cite{Rab63}.  Nowadays
probabilistic automata are widely used in automated verification,
natural-language processing, and machine learning.  In this paper we
consider weighted automata over the ring $(\Q,+,\cdot,0,1)$, which
generalise probabilistic automata.  Note that we restrict to rational
transition weights to permit effective representation of
automata.

Two $\Q$-weighted automata are said to be equivalent if they assign
the same weight to any given word.  It has been shown by
Schutzenberger~\cite{Schutzenberger} and later by Tzeng~\cite{Tzeng}
that equivalence for $\mathbb{Q}$-weighted automata is decidable in
polynomial time.  By contrast, the natural analog of language
inclusion, that one automaton accepts each word with weight at least
as great as another automaton, is undecidable~\cite{CL89}.  Let us
emphasize that we consider the standard ring structure on $\Q$.  For
example, for weighted automata over the max-plus semiring on $\Q$,
equivalence is undecidable~\cite{ABK11,Krob94}.

In this paper we show that the equivalence problem for
$\mathbb{Q}$-weighted automata, and various extensions thereof,
can be efficiently solved by techniques rooted in \emph{polynomial
  identity testing}.  We focus on establishing bounds involving
complexity classes within the class \textbf{P} of polynomial-time
solvable problems.  In particular, we consider the class \textbf{NC} of
problems solvable in polylogarithmic parallel time with polynomially
many processors~\cite{GHR95} (see Section~\ref{sec:prelim} for
background on complexity theory).

It has long been known that equivalence for $\mathbb{Q}$-weighted
automata can be solved in polynomial time~\cite{Schutzenberger,Tzeng}.
There is moreover an \textbf{NC} algorithm for solving
equivalence~\cite{Tzeng96}.  Our first contribution, in
Section~\ref{sec:rand}, is a randomised \textbf{NC} algorithm for
deciding equivalence, based on polynomial identity testing.  The
advantage of using randomisation in this context is that our algorithm
has much lower processor complexity than~\cite{Tzeng96}.  The latter
performs quadratically more work than the classical sequential
procedure.  On the other hand, our randomised algorithm compared well
with the classical sequential algorithm of~\cite{Schutzenberger,Tzeng}
on a collection of benchmarks~\cite{Cav11}.

We also show that our algorithm can be used not just to decide
equivalence but also to generate counterexamples in case of
inequivalence.  However the counterexample generation is essentially
sequential.  We address this deficiency by giving a second randomised
\textbf{NC} algorithm to decide equivalence of automata and output
counterexamples in case of inequivalence.  The algorithm is based on
the Isolation Lemma, a classical technique in randomised algorithms
that has previously been used, e.g., to derive randomised \textbf{NC}
algorithms for matching in graphs~\cite{MVV87}.  Whether there is a
deterministic \textbf{NC} algorithm that outputs counterexamples in
case of inequivalence remains open.

A $\mathbb{Q}$-weighted automaton is \emph{minimal} if no
equivalent automaton has fewer states.  Minimal automata are unique up
to change of basis.  In Section~\ref{sub:minimisation} we give an
\textbf{NC} procedure to decide if a given automaton is minimal.  For
the associated function problem, that of minimising a given automaton,
we give a randomised \textbf{NC} procedure.  Thus the situation for
minimisation is similar to that for equivalence: the decision problem
is in \textbf{NC} whereas the function problem can only be shown to be
in \textbf{RNC}.

In Section~\ref{sec:commutative} we consider probabilistic automata
with rewards on transitions, which can be seen as partially observable
Markov decision processes.  Rewards (and costs, which can be
considered as negative rewards) are omnipresent in probabilistic
modelling for capturing quantitative effects of probabilistic
computations, such as consumption of time, allocation of memory,
energy usage, etc.  For these automata we consider a notion of
\emph{expectation equivalence}, requiring that two automata have the
same expected reward on each input word, and a stronger notion of
\emph{distribution equivalence}, requiring that each word induce the
same distribution on rewards in both automata.  In both cases we give
decision procedures for equivalence by reduction to the case of
$\mathbb{Q}$-weighted automata, thus inheriting the complexity bounds
established there.

We present a case study in which costs are used to model the
computation time required by an RSA encryption algorithm, and show
that the vulnerability of the algorithm to timing attacks depends on
the equivalence of associated probabilistic reward automata.
In~\cite{Kocher96} two possible defenses against such timing leaks
were suggested.  We also analyse their effectiveness.

In Section~\ref{sec:vpa} we consider pushdown automata.  Probabilistic
pushdown automata are a natural model of recursive probabilistic
procedures, stochastic grammars and branching
processes~\cite{EY09,KEM06}.  The equivalence problem for
deterministic pushdown automata has been extensively
studied~\cite{Senizergues97,Stirling02}.  We study the equivalence
problem for \emph{$\mathbb{Q}$-weighted visibly pushdown automata
  (VPA)}~\cite{AM04}.  In a visibly pushdown automaton the stack
operation of a given transition---whether to pop or push---is
determined by the input symbol being read.

We show that the equivalence problem for $\mathbb{Q}$-weighted VPA is
logspace equivalent to \emph{Arithmetic Circuit Identity Testing
  (\textbf{ACIT})}, which is the problem of determining equivalence of
polynomials presented via arithmetic circuits~\cite{ABKM09}.  Several
polynomial-time randomized algorithms are known for \textbf{ACIT}, but
it is a major open problem whether it can be solved in polynomial time
by a deterministic algorithm.  A closely related result is that of
Seidl~\cite{Seidl90}, that equivalence of $\Q$-weighted tree automata
is decidable in randomised polynomial time.  However~\cite{Seidl90}
does not establish a connection with \textbf{ACIT} in either
direction.


%

\section{Preliminaries} \label{sec:prelim}
\subsection{Complexity Classes}
Recall that \textbf{NC} is the subclass of \textbf{P} comprising those
problems considered efficiently parallelisable.  \textbf{NC} can be
defined via \emph{parallel random-access machines (PRAMs)}, which
consist of a set of processors communicating through a shared memory.
A problem is in \textbf{NC} if it can be solved in time $(\log
n)^{O(1)}$ (polylogarithmic time) on a PRAM with $n^{O(1)}$
(polynomially many) processors.  A more abstract definition of
\textbf{NC} is as the class of languages which have \textbf{L}-uniform
Boolean circuits of polylogarithmic depth and polynomial size.  More
specifically, denote by $\mathbf{NC}^k$ the class of languages which
have circuits of depth $O(\log^k n)$.  The complexity class
\textbf{RNC} consists of those languages with randomized \textbf{NC}
algorithms.  We have the following chain of inclusions, none of which
is known to be strict:
\[ \mathbf{NC}^1 \subseteq \mathbf{L} \subseteq \mathbf{NL}
   \subseteq \mathbf{NC}^2 \subseteq \mathbf{NC} \subseteq \mathbf{RNC} 
\cap \mathbf{P} \subseteq \mathbf{P}
   \, .\]
We also have $\mathbf{NC}^k \subseteq \mathbf{SPACE}(O(\log^k n))$,
that is, problems in \textbf{NC} are solvable in polylogarithmic space.

Problems in \textbf{NC} include reachability in directed graphs,
computing the rank and determinant of an integer matrix, solving
linear systems of equations, and the Tree Isomorphism problem.
Problems that are \textbf{P}-hard under logspace reductions include
Circuit Value and Max Flow.  Such problems are not in \textbf{NC}
unless $\mathbf{P}=\textbf{NC}$.  Problems in $\mathbf{RNC}\cap
\mathbf{P}$ include matching in graphs and max flow in $0/1$-valued
networks.  In both cases these problems have resisted classification
as either being in \textbf{NC} or \textbf{P}-hard.  See~\cite{GHR95}
for more details about \textbf{NC} and \textbf{RNC}.
\subsection{Linear Algebra}
\label{sub:linear-algebra}
Given an $m \times n$ matrix $A=(a_{ij})$ and a $k\times l$ matrix
$B=(b_{ij})$, the \emph{Kronecker product} $A \otimes B$ is an $km
\times nl$ matrix defined by
\[ A \otimes B = \left[
 \begin{array}{ccc}
  a_{11}B & \cdots & a_{1n}B \\
  \vdots &         & \vdots \\
  a_{m1}B & \cdots & a_{mn} B
 \end{array}\right ] \]
The following is a key property of the \emph{Kronecker product}:
\begin{proposition}
$(A\otimes B)(C \otimes D) = (AC \otimes BD)$ for
matrices $A,B,C,D$ of appropriate dimensions.
\label{prop:kronecker}
\end{proposition}

Given two $m \times n$ matrices $A=(a_{ij})$ and $B=(b_{ij})$, the
\emph{Hadamard product} $C=A \odot B$ is the $m\times n$ matrix
defined by $c_{ij}=a_{ij}b_{ij}$.

\subsection{Laurent Polynomials}
A \emph{Laurent polynomial} in variables $t_1,\ldots,t_n$ with
coefficients in $\mathbb{Q}$ is an expression of the form $p =
\sum_{i\in I} a_i t_1^{i_1}\dots t_n^{i_n}$, where $I \subseteq
\mathbb{Z}^n$ is a finite set and $a_i \in \mathbb{Q}$.  We say that
$p$ has degree bound $d$ if $|i_1|+\ldots+|i_n| \leq d$.  We write
$\mathbb{Q}[t_1,t_1^{-1},\ldots,t_n,t_n^{-1}]$ for the ring of such
polynomials, with the usual addition and multiplication operations; we
furthermore write $\mathbb{Q}(t_1,t_1^{-1},\ldots,t_n,t_n^{-1})$ for
the corresponding field of fractions, whose elements are quotients of
Laurent polynomials.

The following proposition immediately follows from the cofactor
formula for matrix inversion.
\begin{proposition}
Let $M$ be an $m\times m$ matrix with entries in
$\mathbb{Q}[t_1,t_1^{-1},\ldots,t_n,t_n^{-1}]$ of degree bound $d$.
If $\mathrm{det}(I-M)\neq 0$, then $I-M$ is invertible over
$\mathbb{Q}(t_1,t_1^{-1},\ldots,t_n,t_n^{-1})$, and each entry of
$(I-M)^{-1}$ can be represented as the quotient of Laurent
polynomials, each of degree bound at most $md$.
\label{prop:degree-bound}
\end{proposition}
In the situation of Proposition~\ref{prop:degree-bound} we denote
$(I-M)^{-1}$ by $M^*$.

\section{Equivalence of \texorpdfstring{$\Q$}{Q}-Weighted Automata}
\label{sec:rand}

Given a field $(\F,+,\cdot,0,1)$, an \emph{$\F$-weighted automaton}
$\A = (n, \Sigma, M, \valpha, \veta)$ consists of a positive integer
$n \in \N$ representing the number of states, a finite
alphabet~$\Sigma$, a map $M : \Sigma \to \F^{n \times n}$ assigning a
transition matrix to each alphabet symbol, an initial (row) vector
$\valpha \in \F^n$, and a final (column) vector $\veta \in \F^n$.  We
extend $M$ to $\Sigma^*$ as the matrix product $M(\sigma_1 \ldots
\sigma_k) := M(\sigma_1) \cdot \ldots \cdot M(\sigma_k)$.  The
automaton~$\A$ assigns to each word~$w$ a \emph{weight} $\A(w) \in \F$,
where $\A(w) := \valpha M(w) \veta$.  An automaton~$\A$ is said to be
\emph{zero} if $\A(w) = 0$ for all $w \in \Sigma^*$.  Two automata
$\B, \C$ over the same alphabet~$\Sigma$ are said to be
\emph{equivalent} if $\B(w) = \C(w)$ for all $w \in \Sigma^*$.  

Given two automata $\B,\C$ that are to be checked for equivalence, one
can compute an automaton~$\A$ with $\A(w) = \B(w) - \C(w)$ for all $w
\in \Sigma^*$.  Then $\A$ is zero if and only if $\B$ and~$\C$ are
equivalent.  Given $\B = (\ns{\B}, \Sigma, \Ms{\B}, \alphas{\B},
\etas{\B})$ and $\C = (\ns{\C}, \Sigma, \Ms{\C}, \alphas{\C},
\etas{\C})$, set $\A= (n, \Sigma, M, \valpha, \veta)$ with $n :=
\ns{\B} + \ns{\C}$ and\\[-1mm]
 \[
  M(\sigma) := \begin{pmatrix} \Ms{\B}(\sigma) & 0 \\ 0 & \Ms{\C}(\sigma) \end{pmatrix}\,, \qquad
  \valpha := (\alphas{\B}, -\alphas{\C})\,, \qquad
  \veta := \begin{pmatrix} \etas{\B} \\ \etas{\C} \end{pmatrix}\,.
 \]
This reduction allows us to focus on \emph{zeroness}, i.e., the
problem of determining whether a given $\F$-weighted automaton is
zero.  (Since transition weights can be negative, zeroness is not the
same as emptiness of the underlying unweighted automaton.)  Note that
a witness word $w \in \Sigma^*$ against zeroness of~$\A$ is also a
witness against the equivalence of $\B$ and~$\C$.

In the remainder of this section we present two randomised
$\mathbf{NC}^2$ algorithm algorithms for deciding equivalence of
$\mathbb{Q}$-weighted automata.  The following result from~\cite{Tzeng}
immediately implies decidability of testing zeroness, and hence
equivalence, of $\mathbb{Q}$-weighted automata.

\begin{proposition}
Let $\mathbb{F}$ be any field and
$\mathcal{A}=(n,\Sigma,M,\valpha,\veta)$ an
$\mathbb{F}$-weighted automaton.  Then: (i)~$\Span \{ \valpha M(w) : w
\in \Sigma^*\} = \Span \{ \valpha M(w) : w \in \Sigma^{<n} \}$; (ii) if
$\A$ is not equal to the zero automaton then there exists a word $w
\in \Sigma^*$ of length at most $n-1$ such that $\A(w)\neq 0$.
\label{prop:short}
\end{proposition}

\subsection{Algorithm Based on the Schwartz-Zippel Lemma}
\label{sub:cav-algorithm}

By Proposition~\ref{prop:short} a $\mathbb{Q}$-weighted automaton with
$n$ states is zero if and only if its \emph{$n$-bounded language} is
zero, that is, it assigns weight zero to all words of length at most
$n$.  Inspired by the work of Blum, Carter and Wegman on free Boolean
graphs~\cite{BCW80}, we represent the \emph{$n$-bounded language} of an
automaton by a polynomial in which each monomial represents a word and
the coefficient of the monomial represents the weight of the word.  We
thereby reduce the zeroness problem to polynomial identity testing,
for which there are a number of efficient randomised procedures.

Let $\mathcal{A}=(n,\Sigma,M,\valpha,\veta)$
 be a $\mathbb{Q}$-weighted automaton.
We introduce a family of variables
$\vx = \{ x_{\sigma,i} : \sigma \in \Sigma,\; 1 \leq i \leq n\}$
and associate the monomial $x_{w_1,1}x_{w_2,2}\ldots x_{w_k,k}$ with a
word $w = w_1w_2 \ldots w_k$ of length $k \leq n$.
Then we define the polynomial $P(\vx)$ by
\begin{gather}
P(\vx) := \sum_{k=0}^{n-1} \, \sum_{w \in \Sigma^{k}}
\mathcal{A}(w) \cdot x_{w_1,1}x_{w_2,2}\ldots x_{w_k,k} \, .
\label{eq:poly}
\end{gather}
It is immediate from Proposition~\ref{prop:short} that $P(\vx) \equiv
0$ if and only if $\mathcal{A}$ is zero.

To test whether $P(\vx) \equiv 0$ we
select a value for each variable $x_{\sigma,i}$ independently and
uniformly at random from a set of integers of size $K n$, for some
constant $K$.  Clearly if $P(\vx) \equiv 0$
then this yields the value~$0$.
On the other hand, if $P(\vx) \not\equiv 0$ then
$P$ will evaluate to a nonzero value with probability at least
$(K-1)/K$ by the following result of De Millo and Lipton~\cite{DL78},
Schwartz~\cite{Schwartz80} and Zippel~\cite{Zippel79} and the fact
that $P$ has degree~$n-1$.

\begin{theorem}[\cite{DL78,Schwartz80,Zippel79}] 
Let $\mathbb{F}$ be a field and
  $Q(x_1,\ldots,x_n) \in \mathbb{F}[x_1,\ldots,x_n]$ a multivariate
  polynomial of total degree $d$.  Fix a finite set $\mathbb{S}
  \subseteq \mathbb{F}$, and let $r_1,\ldots,r_n$ be chosen
  independently and uniformly at random from $\mathbb{S}$.  Then
\[ \Pr[Q(r_1,\ldots,r_n)=0 \mid Q(x_1,\ldots,x_n) \not\equiv
0] \leq \frac{d}{|\mathbb{S}|} \, .\]
\label{thm:schwartz-zippel}
\end{theorem}

While the number of monomials in $P$ is proportional to $|\Sigma|^n$, i.e., exponential
in $n$, writing
\begin{gather}
P(\vx) = \valpha  \left( \sum_{i=0}^n
\prod_{j=1}^i \sum_{\sigma\in \Sigma} x_{\sigma,j} \cdot M(\sigma) \right) \veta
\label{eq:evaluation}
\end{gather}
it is clear that $P$ can be
evaluated on a particular set of numerical arguments in time
polynomial in $n$.  The formula~\eqref{eq:evaluation} can be evaluated
in a forward direction, starting with the initial state vector
$\valpha$ and post-multiplying by the transition matrices, or in a
backward direction, starting with the final state vector $\veta$
and pre-multiplying by the transition matrices.  In either case we get
a polynomial-time Monte-Carlo algorithm for testing zeroness of
$\mathbb{Q}$-weighted automata.  The backward variant is shown in
Figure~\ref{fig:rand}.

\begin{figure}
\begin{minipage}{\textwidth}
\textbf{Algorithm $\mathbf{ZERO}$}\\[1ex]
\textbf{Input:} Automaton $\A = (n, \Sigma, M, \valpha,\veta)$\\[2ex]
  if $\valpha \veta \ne 0$ \\
  \ind return ``$\valpha \veta = \A(\varepsilon) \ne 0$'' \\
  $\vv$ := $\veta$ \\
  for $i$ from $1$ to~$n$ do \\
  \ind choose a random vector $r \in \{1,2,\ldots,Kn\}^\Sigma$ \\
  \ind $\vv$ := $\sum_{\sigma \in \Sigma} r(\sigma) M(\sigma) \vv$ \\
  \ind if $\valpha \vv \ne 0$ \\
  \indd return ``$\exists w$ with $|w| = i$ such that $\A(w) \ne 0$'' \\
  return ``$\A$ is zero with probability at least $(K-1)/K$''
\end{minipage}
\caption{Algorithm for testing zeroness} \label{fig:rand}
\end{figure}

\begin{figure}
\begin{minipage}{\textwidth}
\textbf{Algorithm $\mathbf{ZERO + CEX}$}\\[1ex]
\textbf{Input:} Automaton $\A = (n, \Sigma, M, \valpha,\veta)$\\[2ex]
  if $\valpha \veta \ne 0$ \\
  \ind return ``$\valpha \veta = \A(\varepsilon) \ne 0$'' \\
  $\vv_{0}$ :=  $\veta$\\
  for $i$ from $1$ to~$n$ do \\
    \ind choose a random vector $r \in \{1,2,\ldots,Kn\}^\Sigma$ \\
    \ind $\vv_{i}$ := $\left ( \sum_{\sigma \in \Sigma} r(\sigma) M(\sigma) \right) \vv_{i-1}$\\
    \ind if $\valpha \vv_i \ne 0$ \\
    \indd $w$ := $\varepsilon$ \\
    \indd $\vu$ := $\valpha$ \\
    \indd for $j$ from $i$ downto~$1$ do \\
    \inddd choose $\sigma \in \Sigma$ with $\vu M(\sigma) \vv_{j-1} \neq 0$ \\
    \inddd $w$ := $w\sigma$ \\
    \inddd $\vu$ := $\vu M(\sigma)$\\
    \indd return ``$\vu \veta = \A(w) \ne 0$''\\
  return ``$\A$ is zero with probability at least $(K-1)/K$''
\end{minipage}
\caption{Algorithm for testing zeroness, with counterexamples \label{fig:rand_ce}}
\end{figure}

The algorithm runs in time $O(n \cdot |M|)$, where $|M|$ is the number of
nonzero entries in all $M(\sigma)$, provided that sparse-matrix
representations are used.  In a set of case studies
this randomised algorithm outperformed deterministic algorithms~\cite{Cav11}.

We can obtain counterexamples from the randomised algorithm by
exploiting the self-reducible structure of the equivalence problem.
We generate counterexamples incrementally, starting with the empty
string and using the randomised algorithm as an oracle to know at each
stage what to choose as the next letter in our counterexample.  For
efficiency reasons it is important to avoid repeatedly running the
randomised algorithm.  In fact, as shown in Figure~\ref{fig:rand_ce},
this can all be made to work with some post-processing following a
single run of the randomised procedure.

To evaluate the polynomial~$P(\vx)$ we substitute a set of randomly
chosen rational values $\vr = \{ r_{\sigma,i} :\sigma \in \Sigma,\; 1
\leq i \leq n\}$ into Equation~\eqref{eq:evaluation}.  Here we
generalize this to a notion of \emph{partial evaluation} $P_w(\vr)$ of
polynomial~$P$ with respect to values $\vr$ and a word $w \in
\Sigma^m$, $m \leq n$.  We define
\begin{gather}
P_w(\vr) = \valpha M(w) \left( \sum_{i=m}^{n}
\prod_{j=m+1}^i \sum_{\sigma\in \Sigma} r_{\sigma,j} \, M(\sigma)
\right) \veta \, .
\label{eq:par-evaluation}
\end{gather}
Notice that $P_\varepsilon(\vr) = P(\vr)$,
where $\varepsilon$ is the empty word, and, at the other extreme,
$P_w(\vr) = \A(w)$ for any word~$w$ of length~$n$.

\begin{proposition}
Suppose that $w \in \Sigma^m$, where $m<n$.  If $P_w(\vr) \neq 0$ then either
$\A(w) \neq 0$ or $P_{w\sigma}(\vr) \neq 0$ for some $\sigma\in
\Sigma$.
\label{prop:c-ex}
\end{proposition}
\proof
We prove the contrapositive: if $\A(w)=0$ and $P_{w\sigma}(\vr) = 0$
for each $\sigma \in \Sigma$, then $P_w(\vr) = 0$.  This immediately
follows from the equation
\begin{gather*}
P_w(\vr) = \A(w) + \sum_{\sigma \in \Sigma} r_{\sigma,m+1} \, P_{w\sigma}(\vr) \, .
\end{gather*}
This equation is established from the definition of $P_w(\vr)$ as follows:
\begin{eqnarray*}
P_w(\vr) & = & \valpha M(w) \left( \sum_{i=m}^{n} \, 
\prod_{j=m+1}^i \, \sum_{\sigma\in \Sigma} r_{\sigma,j} \, M(\sigma)
\right) \veta\\
& = & \mathcal{A}(w) + \valpha M(w) \left( \sum_{i={m+1}}^{n} \,
\prod_{j=m+1}^i \, \sum_{\sigma\in \Sigma} r_{\sigma,j} \, M(\sigma)
\right) \veta\\
& = & \mathcal{A}(w) + \sum_{\sigma\in \Sigma} r_{\sigma,m+1} \, \valpha M(w\sigma)
\left( \sum_{i={m+1}}^{n} \,
\prod_{j=m+2}^i \, \sum_{\sigma\in \Sigma} r_{\sigma,j} \, M(\sigma)
\right) \veta \\
& = & \A(w) + \sum_{\sigma \in \Sigma} r_{\sigma,m+1} \,
P_{w\sigma}(\vr) \, .\rlap{\hbox to 205 pt{\hfil\qEd}}
\end{eqnarray*}

From Proposition~\ref{prop:c-ex} it is clear that the algorithm in
Figure~\ref{fig:rand_ce} generates a counterexample trace given
$\vr$ such that $P(\vr) \neq 0$.

The algorithm in Figure~\ref{fig:rand} can be parallelised, yielding
an $\mathbf{RNC}$ algorithm, as iterated products of matrices can be
computed in~$\mathbf{NC}$.  On the other hand, the algorithm in
Figure~\ref{fig:rand_ce} yields a counterexample, but apparently cannot
be parallelised efficiently because the counterexample is produced
incrementally.

\subsection{Algorithm Based on the Isolating Lemma}

We now develop a randomised $\mathbf{NC}^2$ procedure that can produce
a counterexample in case of inequivalence.  To this end we employ the
Isolating Lemma of Mulmuley, Vazirani and Vazirani~\cite{MVV87}.  We
use this lemma in a very similar way to~\cite{MVV87}, who are
concerned with computing maximum matchings in graphs in \textbf{RNC}.

\begin{lemma}
Let $\mathcal{F}$ be a family of subsets of a set
$\{x_1,\ldots,x_N\}$. Suppose that each element $x_i$ is assigned a
weight $w_i$ chosen independently and uniformly at random from
$\{1,\ldots,2N\}$.  Define the weight of $S \in \mathcal{F}$ to be
$\sum_{x_i \in S}w_i$.  Then the probability that there is a unique
minimum weight set in $\mathcal{F}$ is at least $1/2$.
\end{lemma}

We will apply the Isolating Lemma in conjunction with
Proposition~\ref{prop:short} to decide zeroness of a $\Q$-weighted
automaton~$\A$.  Suppose $\A$ has $n$ states and alphabet $\Sigma$.
Given $\sigma \in \Sigma$ and $1 \leq i \leq n$, choose a weight
$w_{i,\sigma}$ independently and uniformly at random from the set
$\{1,\ldots,2|\Sigma|n\}$.  Define the weight of a word $u =
\sigma_1\ldots \sigma_k$, $k \leq n$, to be $\mathrm{wt}(u) :=
\sum_{i=1}^k w_{i,\sigma_i}$.  (The reader should not confuse this
with the weight~$\A(u)$ assigned to~$u$ by the automaton~$\A$.)  Then
we obtain a univariate polynomial $P$ from automaton~$\A$ as follows:
\[ P(x) = \sum_{k=0}^n \sum_{u \in \Sigma^k} \mathcal{A}(u) x^{\mathrm{wt}(u)} \, .\]

If $\A$ is equivalent to the zero automaton then clearly $P \equiv 0$.
On the other hand, if $\A$ is non-zero, then by
Proposition~\ref{prop:short} the set $\mathcal{F} = \{ u \in \Sigma^{\leq n} :
\A(u) \neq 0 \}$ is non-empty.  Thus there is a unique minimum-weight
word $u \in \mathcal{F}$ with probability at least $1/2$ by the
Isolating Lemma.  In this case $P$ contains the monomial
$x^{\mathrm{wt}(u)}$ with coefficient $\A(u)$ as its
smallest-degree monomial.  Thus $P \not\equiv 0$ with probability at
least $1/2$.

It remains to observe that from the formula
\[ P(x) = \valpha \left( \sum_{i=0}^n \prod_{j=1}^i \sum_{\sigma \in
    \Sigma} M(\sigma) x^{w_{j,\sigma}} \right) \veta\] and the fact
that iterated products of matrices of univariate polynomials can be
computed in $\mathbf{NC}^2$~\cite{Cook85} we obtain an $\mathbf{RNC}$
algorithm for determining zeroness of $\Q$-weighted automata.

It is straightforward to extend the above algorithm to obtain an
$\mathbf{RNC}$ procedure that not only decides zeroness of $\A$ but
also outputs a word $u$ such that $\A(u) \neq 0$ in case $\A$ is
non-zero. Assume that $\A$ is non-zero and that the random choice of
weights has isolated a unique minimum-weight word $u =
\sigma_1\ldots\sigma_k$ such that $\A(u) \neq 0$.  To determine
whether $\sigma \in \Sigma$ is the $i$-th letter of $u$ we can
increase the weight $w_{i,\sigma}$ by $1$ while leaving all other
weights unchanged and recompute the polynomial $P(x)$.  Then $\sigma$
is the $i$-th letter in $u$ if and only if the minimum-degree monomial
in $P$ changes.  All of these tests can be done independently,
yielding an $\mathbf{RNC}$ procedure.

\begin{theorem}   
Given two $\Q$-weighted automata $\A$ and $\B$, there is an
\textbf{RNC} procedure that determines whether or not $\A$ and $\B$
are equivalent and that outputs a word $w$ with $\A(w) \neq \B(w)$ in
case $\A$ and $\B$ are inequivalent.
\label{thm:rnc-equiv}
\end{theorem}

From a practical perspective, the algorithm is less efficient than those from the previous subsection,
 as it requires computations on univariate polynomials rather than on mere numbers.

\section{Minimisation of \texorpdfstring{$\Q$}{Q}-Weighted Automata}
\label{sub:minimisation}

A $\mathbb{Q}$-weighted automaton is \emph{minimal} if there is no
equivalent automaton with strictly fewer states.  It is known that
minimal automata are unique up to a change of basis~\cite{CarlyleP71}.
In this section we give an \textbf{NC} algorithm to decide whether a
given $\mathbb{Q}$-weighted automaton $\mathcal{A}$ is minimal.  We
also give an \textbf{RNC} algorithm that computes a minimal automaton
equivalent to a given $\mathbb{Q}$-weighted automaton $\mathcal{A}$.

\subsection{Deciding Minimality}

Let $\mathcal{A}=(n,\Sigma,M,\valpha,\veta)$ be an automaton.  Define
the (infinite) matrix $F$ to have rows indexed by $\Sigma^*$ and
columns indexed by $\{1,\ldots,n\}$, with the row indexed by $w \in
\Sigma^*$ being the vector $\valpha M(w)$.  The \emph{forward space}
$\mathsf{F}$ is defined to be the row space of $F$.  Similarly define
the matrix $B$ to have rows indexed by $\{1,\ldots,n\}$ and columns
indexed by $\Sigma^*$, with the column indexed by $w \in \Sigma^*$
being the vector $M(w)\veta$.  The \emph{backward space} $\mathsf{B}$
is defined to be the column space of $B$.  The product $H = FB$
is called the \emph{Hankel matrix}; it has rows and columns indexed by
$\Sigma^*$ with $H_{x,y} = \valpha M(x)M(y)\veta = \mathcal{A}(xy)$.
By linear algebra we have $\mathrm{rank}(H) \leq \min\{
\mathrm{rank}(F),\mathrm{rank}(B) \} \leq n$.  A fundamental
result~\cite{CarlyleP71} is that the above inequalities are tight
precisely when $\mathcal{A}$ is minimal:

\begin{proposition}[Carlyle and Paz]
An automaton $\mathcal{A}$ with $n$ states is minimal if and only if
the Hankel matrix $H$ has rank $n$.
\label{prop:carlyle-paz}
\end{proposition}

Using this result we show

\begin{theorem}
Deciding whether a $\mathbb{Q}$-weighted automaton is minimal is in
\textbf{NC}.
\end{theorem}
\begin{proof}
To check that a given automaton
$\mathcal{A}=(n,\Sigma,M,\valpha,\veta)$ is minimal it suffices to
verify that the associated Hankel matrix $H$ has rank $n$.  Since
$H=FB$, this holds if and only if the matrices $F$ and $B$ both have
rank $n$.  We show how to check that $F$ has rank $n$; the procedure
for $B$ is entirely analogous.

Let $\Ft$ be the sub-matrix of $F$ obtained by retaining only those
rows indexed by words in $\Sigma^{<n}$.  By
Proposition~\ref{prop:short}(i) we have
$\mathrm{rank}(F)=\mathrm{rank}(\Ft)$.  Thus

\begin{eqnarray*}
\mathrm{rank}(F)=n & \Leftrightarrow & \mathrm{rank}(\Ft)=n\\
                   & \Leftrightarrow & \mathrm{ker}(\Ft)=\{0\}\\
                   & \Leftrightarrow & \mathrm{ker}(\Ft^T \Ft) = \{0\}\\
                   & \Leftrightarrow & \mathrm{det}(\Ft^T \Ft) \neq 0 \, .
\end{eqnarray*}
The middle equivalence holds because for any vector $x \in
\mathbb{Q}^n$, $\Ft^T \Ft x =0$ implies $0=x^T \Ft^T \Ft x =
(\Ft x)^T\Ft x$, which in turn implies that $\Ft x=0$.

Since determinants can be computed in \textbf{NC} it only remains to
show that we can compute each entry of the $n \times n$ matrix $\Ft^T\Ft$
in \textbf{NC}.  Let $\ve_i \in \mathbb{Q}^n$ be the column vector with
$1$ in the $i$-th position and $0$ in all other positions.  Given $1
\leq i,j \leq n$ we have
\begin{eqnarray*}
(\Ft^T \Ft)_{ij} & = & \sum_{w \in \Sigma^{<n}} (\valpha M(w) \ve_i)
                                            (\valpha M(w) \ve_j)\\
             & = & \sum_{w \in \Sigma^{<n}} (\valpha \otimes \valpha)
                   (M(w) \otimes M(w)) (\ve_i \otimes \ve_j) \\
             & = & (\valpha \otimes \valpha) \left (
                   \sum_{k=0}^{n-1} \sum_{w \in \Sigma^k} \left (
                   M(w) \otimes M(w) \right ) \right )
                   (\ve_i \otimes \ve_j) \\
             & = & (\valpha \otimes \valpha) \left (
                   \sum_{k=0}^{n-1} \left(
                   \sum_{\sigma \in \Sigma} (M(\sigma) \otimes M(\sigma))
                   \right)^k \right )
                  (\ve_i \otimes \ve_j) \, .
\end{eqnarray*}
But this last expression can be computed in \textbf{NC} since sums and
matrix powers can be computed in \textbf{NC}~\cite{Cook85}.
\end{proof}

\subsection{Minimising an Automaton}

Next we give an \textbf{RNC} algorithm to minimise a given automaton.
The key idea is that we can compute a basis of the forward space
$\mathsf{F}$ by generating random vectors in the space.  We show that
a randomly generated set of such vectors of cardinality equal to the
dimension of $\mathsf{F}$ is likely to be a basis of $\mathsf{F}$. We
can likewise compute a basis of the backward space $\mathsf{B}$.  We
give the construction of the forward space; the proof for the backward
space is similar.

The construction involves an application of polynomial identity
testing in similar manner to Section~\ref{sub:cav-algorithm}.
Consider again a family of variables $\vx = \{ x_{\sigma,i} : \sigma
\in \Sigma,\; 1 \leq i \leq n\}$ and associate the monomial
$x_{w_1,1}x_{w_2,2}\ldots x_{w_k,k}$ with a word $w = w_1w_2 \ldots
w_k$.  Then we define the row vector $\rho(\vx)
\in \mathbb{Q}[\vx]^n$ by
\begin{gather}
\rho(\vx) := \sum_{k=0}^n \, \sum_{w \in \Sigma^{k}}
\valpha M(w) \cdot x_{w_1,1}x_{w_2,2}\ldots x_{w_k,k}
\, .
\end{gather}

Note that evaluating $\rho(\vx)$ at a vector of rationals $\vr = (
r_{\sigma,i} :\sigma \in \Sigma,\; 1 \leq i \leq n)$ yields a vector
$\rho(\vr)$ in the forward space $\mathsf{F}$.

\begin{proposition}
Let $\mathsf{U}$ be a proper subspace of $\mathsf{F}$ and let $K$ be a
positive integer.  Then for $\boldsymbol{r}$ chosen uniformly at
random from $\{1,\ldots,Kn\}^{\Sigma \times n}$ we have
$\Pr(\rho(\boldsymbol{r}) \in \mathsf{U}) \leq 1/K$.
\label{prop:errorbound}
\end{proposition}
\begin{proof}
Pick a non-zero vector $\vv \in \mathsf{F}$ that is orthogonal to
$\mathsf{U}$.  Notice that the polynomial $\rho(\boldsymbol{x})\vv^T$ is
non-zero since the coefficient of the monomial corresponding to a word
$w \in \Sigma^{<n}$ is $\valpha M(w) \vv^T$, and this is clearly non-zero
for at least one $w$.  Now $\rho(\boldsymbol{r}) \in \mathsf{U}$ only if
$\rho(\boldsymbol{r})\vv^T =0$.  Since $\rho(\boldsymbol{x})\vv^T$ has
degree at most $n$, it follows from Theorem~\ref{thm:schwartz-zippel}
that $\Pr(\rho(\boldsymbol{r}) \in \mathsf{U})$ is at most $1/K$.
\end{proof}

The procedure to generate a basis for the forward space $\mathsf{F}$
is shown in Figure~\ref{fig:forward-basis}.
\begin{figure}
\begin{minipage}{\textwidth}
\textbf{Algorithm}\;\textbf{Forward-Basis}\\[1ex]
\textbf{Input:} Automaton $\A = (n, \Sigma, M, \valpha,\veta)$
and error parameter $K$\\[2ex]
   for $i$ from $1$ to~$n$ do \\
  \ind choose a random vector $\vr^{(i)}
       \in \{1,2,\ldots,Kn\}^{\Sigma \times n}$ \\
  \ind $\vv_i$ := $\rho(\vr^{(i)})$\\
let $k$ be maximum such that $\{\vv_1,\ldots,\vv_k\}$
       is linearly independent\\
  return ``$\{\vv_1,\ldots,\vv_k\}$ is a basis of $\mathsf{F}$''
\end{minipage}
\caption{Algorithm for generating a basis of the forward space}
\label{fig:forward-basis}
\end{figure}
The algorithm \textbf{Forward-Basis} necessarily returns a linearly
independent set of vectors in the forward space.  It only fails to
output a basis if $\vv_{m+1} \in \Span\{\vv_1,\ldots,\vv_m\}$ for some
$m < \dim(\mathsf{F})$.  By Proposition~\ref{prop:errorbound} this
happens with probability at most $1/K$ for any given $m$, so the total
probability that \textbf{Forward-Basis} does not give a correct output
is at most $n/K$.  Thus, e.g., choosing $K=3n$ we have an error
probability of at most $1/3$.

It remains to observe that \textbf{Forward-Basis} can be made to run
in $O(\log^2 n)$ parallel time.  We perform the assignments $\vv_i
\mathrel{:=} \rho(\vr^{(i)})$ for $i=1,\ldots,n$ in parallel.  As
observed in Section~\ref{sub:cav-algorithm}, the computation of
$\rho(\vr^{(i)})$ involves an iterated matrix product, which can be
done in $O(\log^2 n)$ parallel time.  We also check linear
independence of $\{\vv_1,\ldots,\vv_k\}$ for $k=1,\ldots,n$ in
parallel.  Each check involves computing the rank of an $k\times n$
matrix, which can again be done in $O(\log^2 n)$ parallel
time~\cite{IbarraMR80}.

Given bases of $\mathsf{F}$ and $\mathsf{B}$, minimisation proceeds
via a classical construction of Sch\"{u}tzen\-berger~\cite{Schutzenberger}.
We briefly recall this construction and show that it
can be implemented in \textbf{NC} by making one call to algorithm
\textbf{Forward-Basis} and one call to the corresponding backward
version of this algorithm.

Let $\nF \in \N$ and $\FF \in \Q^{\nF \times n}$ be such that the rows
of~$\FF$ form a basis of the forward space~$\FV$, with the first row of
$\FF$ being $\valpha$.  Similarly, let $\nB \in \N$ and $\BB \in \Q^{n
  \times \nB}$ be such that the columns of~$\BB$ form a basis of the
backward space~$\BV$, with the first column of $\BB$ being $\veta$.
Since $\FV M(\sigma) \subseteq \FV$ and $M(\sigma) \BV \subseteq \BV$
for all $\sigma \in \Sigma$, there exist maps $\MF : \Sigma \to
\Q^{\nF \times \nF}$ and $\MB : \Sigma \to \Q^{\nB \times \nB}$ such
that
 \begin{equation}
  \FF M(\sigma) = \MF(\sigma) \FF \quad \text{and} \quad
  M(\sigma) \BB = \BB \MB(\sigma) \quad \text{for all $\sigma \in \Sigma$.}
\label{eq:equations}
 \end{equation}
Call $\AF := (\nF, \Sigma, \MF, \ve_1, \FF \veta)$ a \emph{forward
  reduction} of~$\A$ with base~$\FF$ and similarly $\AB := (\nB, \Sigma, \MB,
\valpha \BB, \ve_1^T)$ a \emph{backward reduction} of~$\A$ with base~$\BB$.


\begin{proposition}[~\cite{Schutzenberger}]
 Let $\A$ be an automaton.  Then $\ABF$ is minimal and equivalent
 to~$\A$.
\end{proposition}
%

\begin{theorem}
There is an \textbf{RNC} algorithm that transforms a given
automaton into an equivalent minimal automaton.
\end{theorem}
\begin{proof}
Let $\mathcal{A}=(n,\Sigma,M,\valpha,\veta)$ be an automaton.  We have
already shown that we can compute in randomised \textbf{NC} a
matrix $\FF$ whose rows form a basis of the forward space of $\A$.
Given $\FF$ we can compute the forward reduction $\AF$ in \textbf{NC}
since each transition matrix $\MF(\sigma)$ is uniquely defined as the
solution to the linear system of equations (\ref{eq:equations}).
Using the same reasoning we can compute $\ABF$ from $\AF$ in
randomised \textbf{NC}.  This is the minimal automaton that we seek.
\end{proof}

\section{Probabilistic Reward Automata} \label{sec:commutative}
In this section we consider \emph{probabilistic reward automata},
which extend Rabin's probabilistic automata~\cite{Rab63} with rewards
on transitions.  The resulting notion can be seen as a type of
partially observable Markov Decision Process~\cite{Bellman57}.  A
similar model has been investigated from the point of view of language
theory in~\cite{ChatterjeeDH09}.  Rewards are allowed to be negative,
in which case they can be seen as costs.  In Example~\ref{ex:rsa} we
use costs to record the passage of time in an encryption protocol.

A \emph{Probabilistic Reward Automaton} is a tuple $\A =
(n,s,\Sigma,M,R,\valpha,\veta)$, where $n \in \N$ is the number of
states; $s \in \N$ is the number of types of reward; $\Sigma$ is a
finite alphabet, $M(\sigma)$ is an $n \times n$ rational sub-stochastic
matrix for each $\sigma \in \Sigma$; $R(\sigma)$ is an $n \times n$
matrix with entries in $\{-1,0,1\}^s$ for each $\sigma \in \Sigma$;
$\valpha$ is an $n$-dimensional rational stochastic row vector;
$\veta$ is a rational $n$-dimensional column vector with all entries
lying in the interval $[0,1]$.  We think of $M(\sigma)$ as the
transition matrix, $R(\sigma)$ as the reward matrix, $\valpha$ as the
initial-state vector, and $\veta$ as the final-state vector.

The total reward of a run is the sum of the rewards along all its
transitions.  The expected reward of a word is the sum of the rewards
of all runs over that word, weighted by their respective probabilities.
Formally, given a word $w = w_1,\ldots,w_k$ and a path of states $p =
p_0,\ldots,p_k$, the probability and total reward of the path are
respectively defined by
\[
\Pr(p) \, = \, \valpha_{p_0} \left( \prod_{i=1}^{k}
                          M(\sigma)_{p_{i-1},p_i} \right ) \veta_{p_k}
\quad\mbox{and}\quad 
\mathrm{Reward}(p) \, = \, \sum_{i=1}^{k} R(w_i)_{p_{i-1},p_{i}} \, .
\]

The \emph{value} of the word $w$ is the expected reward over all runs:
\begin{gather}
\A(w) \, = \, \sum_{p \in \{1,\ldots,n\}^{k+1}}
\Pr(p) \cdot \mathrm{Reward}(p) \, .
\label{eq:average-reward}
\end{gather}

\subsection{Expectation Equivalence}
\label{sub:exp-equiv}
Two probabilistic reward automata $\A$ and $\B$ over the same alphabet
$\Sigma$ are defined to be \emph{equivalent in expectation} if
$\A(w)=\B(w)$ for all words $w \in \Sigma^*$.  In this section we give
a simple reduction of the equivalence problem for probabilistic reward
automata to the equivalence problem for $\mathbb{Q}$-weighted
automata.  The idea is to combine transition probabilities and rewards
in a single matrix.  Without loss of generality we consider automata
with a single type of reward; the general problem can be reduced to
this by considering each component separately.

Let $\A = (n,\Sigma,M,R,\valpha,\veta)$ be a probabilistic reward
automaton.  We define a $\mathbb{Q}$-weighted automaton $\B =
(2n,\Sigma,M',\valpha',\veta')$ such that $\A(w)=\B(w)$ for each word
$w \in \Sigma^*$.  First we introduce the following matrices:
\[ A = \left[\begin{array}{cc} 1 & 0
  \end{array}\right]
  \qquad E = \left[\begin{array}{c} 0\\1
  \end{array}\right] \qquad C = \left[ \begin{array}{cc} 0 & 1 \\
                         0 & 0
  \end{array}\right] \]
We also write $I_n$ for the $n\times n$ identity matrix.
Now we define
\begin{eqnarray*}
\valpha'  & := & \valpha \otimes A\\
   \veta' & := & \veta \otimes  E\\
M'(\sigma) & := & (M(\sigma) \otimes I_2) + ((M(\sigma) \odot R(\sigma)) \otimes C)
\end{eqnarray*}
where $\otimes$ denotes Kronecker product and $\odot$ denotes Hadamard
product (cf.~Section~\ref{sub:linear-algebra}).

\begin{proposition}
$\A(w)=\B(w)$ for all words $w \in \Sigma^*$.
\label{prop:reduction}
\end{proposition}
\begin{proof}
We show by induction that for all words $w \in \Sigma^*$ we have
\begin{gather}
M'(w) = (M(w) \otimes I_2) +
\Bigg(\sum_{\substack{w',w''\\w=w' a w''}} (M(w')(M(a) \odot R(\sigma))M(w'')) \otimes C \Bigg) \, .
\label{eq:IH}
\end{gather}
The base case, $w = \varepsilon$, is clear.  For the induction step we have
\begin{eqnarray*}
M'(w\sigma) & = & M'(w)M'(\sigma) \\
            & = & (M(w)\otimes I_2)(M(\sigma)\otimes I_2)+
                  (M(w)\otimes I_2)((M(\sigma)\odot R(\sigma))\otimes C)\\
            &   & + \Bigg(\sum_{\substack{w',w''\\w=w' a w''}} (M(w')(M(a) \odot R(\sigma))M(w'')) \otimes C\Bigg)
                    (M(\sigma)\otimes I_2)\\
            &  &  + \Bigg(\sum_{\substack{w',w''\\w=w' a w''}} (M(w')(M(a) \odot R(\sigma))M(w'')) \otimes C\Bigg)
                  ((M(\sigma)\odot R(\sigma))\otimes C)
\end{eqnarray*}
But using Proposition~\ref{prop:kronecker} and the identity $C^2=0$, the above expression simplifies to
\[ (M(w\sigma)\otimes I_2) +
\Bigg(\sum_{\substack{w',w''\\w\sigma =w' a w''}} (M(w')(M(a) \odot R(\sigma))M(w'')) \otimes C\Bigg) \, .\]
This completes the induction step.

Using Proposition~\ref{prop:kronecker} and the fact that $AE=0$ and $ACE=I_1$ it follows from (\ref{eq:IH}) that
\begin{align*}
\B(w) = \valpha' M'(w) \veta'  & =
        \sum_{\substack{w',w''\\w=w' a w''}} \valpha(M(w')(M(a)\odot R(a))M(w''))\veta\\
   & = \sum_{i=1}^k \sum_{p \in \{1,\ldots,n\}^{k+1}} \alpha_{p_0} \left( \prod_{j=1}^k M(w_j)_{p_{j-1},p_j} \, R(w_i)_{p_{i-1},p_i} \right) \eta_{p_k}\, .
\end{align*}
But the equivalence of the above expression and
(\ref{eq:average-reward}) follows from distributivity of
multiplication over addition.
\end{proof}

\begin{corollary}
Expectation equivalence of probabilistic reward automata can be
decided in \textbf{NC}.  Moreover there is an \textbf{RNC} procedure
that determines whether or not two automata are equivalent and outputs
a word on which they differ in case they are inequivalent.
\end{corollary}
\begin{proof}
The first part follows by combining Proposition~\ref{prop:reduction}
with the \textbf{NC} algorithm for $\mathbb{Q}$-weighted automaton
equivalence in~\cite{Tzeng96}.  The second part follows by combining
Proposition~\ref{prop:reduction} with Theorem~\ref{thm:rnc-equiv}.
\end{proof}

\subsection{Distribution Equivalence}

Two probabilistic reward automata are called \emph{distribution equivalent}
if they induce identical distributions on rewards for each input word
$w \in \Sigma^*$.  We formalise this notion by translating
probabilistic reward automata into $\Q$-weighted automata over the field
$\mathbb{F}=\mathbb{Q}(t_1,t_1^{-1},\ldots,t_s,t_s^{-1})$ of rational
Laurent functions, as defined in Section~\ref{sec:prelim}.  We
consider $\varepsilon$-transitions in this section because they are
convenient for applications (cf.~Example~\ref{ex:geom}) and because we
cannot rely on existing $\varepsilon$-elimination results in the
presence of rewards.

Let $\A=(n,s,\Sigma,M,R,\valpha,\veta)$ be a probabilistic reward
automaton, where $\varepsilon \in \Sigma$.  To make
$\varepsilon$-elimination more straightforward, we assume that the
transition matrix $M(\varepsilon)$ has no recurrent states, i.e., that
its spectral radius is strictly less than one.  We now define an
$\mathbb{F}$-weighted automaton $\A' = (n,\Sigma,M',\valpha,\veta)$ as
follows.  For $1 \leq i,j \leq n$, let $M'(\sigma)_{i,j} =
at_1^{k_1}\ldots,t_s^{k_s}$, where $M(\sigma)_{i,j}=a$ and
$R(\sigma)_{i,j}=(k_1,\ldots,k_s)$.
We extend $M'$ to a map $M' : \Sigma^* \to \mathbb{F}^{n \times n}$ by
defining
\begin{gather}
M'(w) := M'(\varepsilon)^* M'(w_1) M'(\varepsilon)^*
           \cdots M'(w_m) M'(\varepsilon)^* \, 
\label{eq:prod-def}
\end{gather}
for a word $w=w_1\ldots w_m$.  Our convention on
$\varepsilon$-transitions implies that
$\mathrm{det}(I-M'(\varepsilon)) \neq 0$ and therefore, by
Proposition~\ref{prop:degree-bound}, that $M'(\varepsilon)^*$ is
well-defined and has entries whose numerators and denominators are
Laurent polynomials with degree bound $sn$.  It follows that the
entries of $M'(w)$ have degree bound $(sn+1)m$.

Two probabilistic reward automata $\B, \C$ over the same
alphabet~$\Sigma$ and with the same number of reward types are said to
be \emph{equivalent} if the corresponding $\mathbb{F}$-weighted
automata $\B'$ and $\C'$ are equivalent, i.e., $\B'(w)=\C'(w)$ for all
words $w \in \Sigma^*$.  Now Proposition~\ref{prop:short} implies that
equivalence for $\mathbb{F}$-weighted automata is decidable, but the
algorithms of Sch\"{u}tzenberger~\cite{Schutzenberger} and
Tzeng~\cite{Tzeng} do not yield polynomial-time procedures in our case
because the complexity of solving systems of linear equations over the
field $\mathbb{Q}(t_1,t_1^{-1},\ldots,t_s,t_s^{-1})$ is not polynomial
in $s$ (indeed the solution need not have length exponential in $s$).
However, it not difficult to give a randomised polynomial-time
algorithm to decide equivalence of probabilistic reward automata.

Let $\A'$ be the $\mathbb{F}$-weighted automaton corresponding to a
probabilistic reward automaton $\A$ with $n$ states.  For each word $w
\in \Sigma^*$ of length at most $n$ we have a rational function
$\A'(w)$ whose numerator and denominator are polynomials of degree at
most $d:=(sn+1)n$, as observed above.  Now consider the set
$R:=\{1,2,\ldots,2d\}^s$.  Suppose that we pick $\vr \in R$ uniformly
at random.  Denote by $\A'(w)(\vr)$ the result of substituting $\vr$
for the formal variables $t_1,\ldots,t_s$ in the rational function
$\A'(w)$.  Clearly if $\A'$ is a zero automaton then $\A'(w)(\vr)=0$
for all $\vr\in R$.  On the other hand, if $\A'$ is non-zero then by
Proposition~\ref{prop:short} there exists a word $w\in\Sigma^*$ of
length at most $n$ such that $\A'(w)\not\equiv 0$.  Since the degree
of the rational expression $A'(w)$ is at most $d$ it follows from the
Schwartz-Zippel theorem~\cite{DL78,Schwartz80,Zippel79} that the
probability that $\A(w)(\vr)=0$ is at most $1/2$.

Thus our randomised procedure is to pick $\vr \in R$ uniformly at
random and to check whether $\A(w)(\vr)=0$ for some $w \in \Sigma^*$.
To perform this final check we show that there is a $\Q$-weighted
automaton $\B$ such that $\A'(w)(\vr)=\B(w)$ for all $w \in \Sigma^*$.
Then check $\B$ for zeroness using, e.g., Tzeng's
algorithm~\cite{Tzeng}.  The automaton $\B$ has the form $\B =
(\ns{\B}, \Sigma, \Ms{\B}, \alphas{\B}, \etas{\B})$, where
$\ns{\B}=n$, $\alphas{\B}=\valpha$, $\etas{\B}=\veta$ and
$\Ms{\B}(\sigma) = M(\sigma)(\vr)$ for all $\sigma \in \Sigma$.

\begin{theorem}
There is an \textbf{RNC} procedure that determines whether or not two
probabilistic reward automata are distribution equivalent, and which
outputs a word on which they differ in case they are inequivalent.
\end{theorem}

\begin{example}
We consider probabilistic programs that randomly increase and decrease
a single counter (initialised with~$0$) so that upon termination the
counter has a random value~$X \in \Z$.  The programs should be such
that $X$ is a random variable with $X = Y-Z$ where $Y$ and~$Z$ are
independent random variables with a geometric distribution with
parameters $p = 1/2$ and $p=1/3$, respectively.  (By that we mean that
$\Pr(Y=k) = (1-p)^k p$ for $k \in \{0, 1, \ldots\}$, and similarly for
$Z$.)  Figure~\ref{fig:ex-geom-dist-code} shows code in the syntax of
the \apex{} tool~\cite{12KMOWW:CAV}.
\begin{figure}
\begin{center}
\begin{tabular}{cc}
\begin{minipage}[t]{0.5\linewidth}
\begin{verbatim}
inc:com, dec:com |-
  var%2 flip;
  flip := 0;
  while (flip = 0) do {
    flip := coin[0:1/2,1:1/2];
    if (flip = 0) then {
     inc;
    };
  };
  flip := 0;
  while (flip = 0) do {
    flip := coin[0:2/3,1:1/3];
    if (flip = 0) then {
     dec;
    };
  }
:com
\end{verbatim}
\end{minipage}
&

\begin{minipage}[t]{0.5\linewidth}
\begin{verbatim}
inc:com, dec:com |-
  var%2 flip;
  flip := coin[0:1/2,1:1/2];
  if (flip = 0) then {
    while (flip = 0) do {
      flip := coin[0:1/2,1:1/2];
      if (flip = 0) then {
       inc;
      };
    };
  } else {
    flip := 0;
    while (flip = 0) do {
      dec;
      flip := coin[0:2/3,1:1/3];
    };
  }
:com
\end{verbatim}
\end{minipage}
\end{tabular}
\end{center}
\caption{Two \apex{} programs for producing a counter that is distributed as the difference between two geometrically distributed random variables.}
\label{fig:ex-geom-dist-code}
\end{figure}

The program on the left consecutively runs two while loops: it first
increments the counter according to a geometric distribution with
parameter~$1/2$ and then decrements the counter according to a
geometric distribution with parameter~$1/3$, so that the final counter
value is distributed as desired.  The program on the right is more
efficient in that it runs only one of two while loops, depending on a
single coin flip at the beginning.  It may not be obvious though that
the final counter value follows the same distribution as in the left
program.  We used the \apex{} tool to translate the programs to the
probabilistic reward automata $\B$ and~$\C$ shown in
Figure~\ref{fig:ex-geom-dist-aut}.  Here each counter increment
corresponds to a reward of $1$ and each counter decrement to a reward
of $-1$.
\begin{figure}
 \begin{tabular}{c@{\hspace{20mm}}c}
 \begin{tikzpicture}[shorten >=1pt,baseline=0]
  \node[state] (1) at (0,0) {$1, \frac16$};
  \node[state] (2) at (3,0) {$2, \frac13$};
  \draw[->] (-0.8,0) -- (1);
  \draw[->] (1) .. controls +(250:1.4cm) and +(290:1.4cm) .. node[above] {$\varepsilon$} node[below,yshift=1mm] {$\frac12$ : inc} (1);
  \draw[->] (2) .. controls +(250:1.4cm) and +(290:1.4cm) .. node[above] {$\varepsilon$} node[below,yshift=1mm] {$\frac23$ : dec} (2);
  \draw[->] (1) to node[above] {$\varepsilon$} node[below] {$\frac13$ : dec} (2);
 \end{tikzpicture}
 &
 \begin{tikzpicture}[shorten >=1pt,baseline=0]
  \node[state] (1) at (0,0) {$1, \frac14$};
  \node[state] (2) at (3,1) {$2, \frac12$};
  \node[state] (3) at (3,-1) {$3, \frac13$};
  \draw[->] (-0.8,0) -- (1);
  \draw[->] (1) to node[above] {$\varepsilon$} node[below,sloped] {$\frac14$ : inc} (2);
  \draw[->] (2) .. controls +(110:1.4cm) and +(70:1.4cm) .. node[above] {$\varepsilon$} node[below,yshift=0.5mm] {$\frac12$ : inc} (2);
  \draw[->] (1) to node[above] {$\varepsilon$} node[below,sloped] {$\frac12$ : dec} (3);
  \draw[->] (3) .. controls +(250:1.4cm) and +(290:1.4cm) .. node[above] {$\varepsilon$} node[below,yshift=0.5mm] {$\frac23$ : dec} (3);
 \end{tikzpicture}
 \\
 ($\B$) & ($\C$)
 \end{tabular}
\caption{Automata produced from the code in
  Figure~\ref{fig:ex-geom-dist-code}.  The states are labelled with
  their number and their ``acceptance probability'' ($\veta$-weight).
  In both automata, state~1 is the only initial state ($\valpha_1 = 1$
  and $\valpha_i = 0$ for $i \ne 1$).  The transitions are labelled
  with the input symbol~$\varepsilon$, with a probability (weight) and
  a cost.}
\label{fig:ex-geom-dist-aut}
\end{figure}
%
Since the input alphabets are empty,
 it suffices to consider the input word~$\varepsilon$ when comparing $\B$ and~$\C$ for equivalence.
If we construct the difference automaton $\A = (5,1,\emptyset,M,\valpha,\veta)$
 and invert the matrix of polynomials $I-M(\varepsilon)$, we obtain
 \[
  \A(\varepsilon)(x) =
    \left(\frac{2}{x-2}, \frac{2}{(3x-2)(x-2)}, 1, \frac{-x}{2(x-2)}, \frac{3}{2(3x-2)}  \right) \veta \equiv 0\,,
 \]
 which proves equivalence of $\B$ and~$\C$.
Notice that the actual algorithm would not compute $\A(\varepsilon)(x)$ as a polynomial,
 but it would compute $\A(\varepsilon)(r)$ only for a few concrete values $r \in \Q$.
\label{ex:geom}
\end{example}

\begin{example}
RSA~\cite{RSA} is a widely-used cryptographic algorithm.  Popular
implementations of the RSA algorithm have been shown to be vulnerable
to timing attacks that reveal private keys~\cite{Kocher96,BrumleyB05}.
The preferred countermeasures are blinding techniques that randomise
certain aspects of the computation, which are described in,
e.g.,~\cite{Kocher96}.  We model the timing behaviour of the RSA
algorithm using probabilistic cost automata, where costs encode
time. These automata are produced by \apex{}, and are then used to
check for timing leaks with and without blinding.

At the heart of RSA decryption is a modular exponentiation,
which computes the value $m^d \mod N$ where $m\in \{0,\ldots,N-1\}$ is the encrypted message, $d\in \mathbb N$ is the private
decryption exponent and $N\in \mathbb N$ is a modulus. An attacker wants to find out $d$.
We model RSA decryption in \apex{} by implementing modular exponentiation
by iterative squaring (see Figure~\ref{fig:code-rsa}).
We consider the situation where the attacker is able
to control the message $m$, and tries to derive $d$ by observing the runtime distribution
over different messages $m$.
Following~\cite{Kocher96} we assume that the running time of multiplication depends on the operand values
(because a source-level multiplication typically corresponds to a cascade of processor-level multiplications).
By choosing the `right' input message $m$, an attacker can observe which private keys are most likely.

We consider two blinding techniques mentioned in Kocher~\cite{Kocher96}.
The first one is base blinding, i.e., the message is multiplied by $r^d$ before exponentiation where $d$ is a random number,
which gives a result that can be fixed by dividing by $r$ but makes it impossible for the attacker to control the
basis of the exponentiation.
The second one is exponent blinding, which adds a multiple of the group order $\varphi(N)$ of $\mathbb Z / N \mathbb Z$
to the exponent, which doesn't change the result of the exponentiation\footnote{Euler's totient function $\varphi$ satisfies $a^{\varphi(N)} \equiv 1 \mod N$ for all $a\in \mathbb Z$.}
but changes the timing behaviour.

Figure~\ref{fig:ex-rsa} shows the automaton
for $N=10$, and private key $0,1,0,1$ with message blinding enabled.
The \apex{} program is given in Figure~\ref{fig:code-rsa}.

We investigate the effectiveness of blinding.
Two private keys are indistinguishable if the resulting automata are equivalent.
The more keys are indistinguishable the safer the algorithm.
We analyse which private keys are identified by plain RSA, RSA with a blinded
message and RSA with blinded exponent.

For example, in plain RSA, the following keys $0,1,0,1$ and $1,0,0,1$
are indistinguishable, keys $0,1,1,0$ and $0,0,1,1$ are
indistinguishable with base blinding, lastly $1,0,0,1$ and $1,0,1,1$
are equivalent only with exponent blinding.  Overall 9 different keys
are distinguishable with plain RSA, 7 classes with base blinding and 4
classes with exponent blinding.
\begin{figure}
  \begin{center}
\begin{verbatim}
const N := 10;    // modulus
const Bits := 4 ; // number of bits of the key

m :int%N, inc:com |-
var%2 exponent[Bits] = [0,1,0,1];
com power(x:int%N) {
   var%N s := 1;
   var%N R;
   for(var%(Bits + 1) k; k < Bits; ++k) do {
      R:=s;
      if(exponent[k]) then {
         R := R*x;
         if(5<=R) then { inc; inc } else { inc }
      }
      s := R*R;
   }
}
var%N message := m*rand[N]; // blinding
power(message) : com
\end{verbatim}
\end{center}
\caption{\apex{} code for RSA.}
\label{fig:code-rsa}
\end{figure}
\begin{figure}
  \begin{center}
  \begin{tabular}{c}
  \includegraphics[scale=0.75, angle=270]{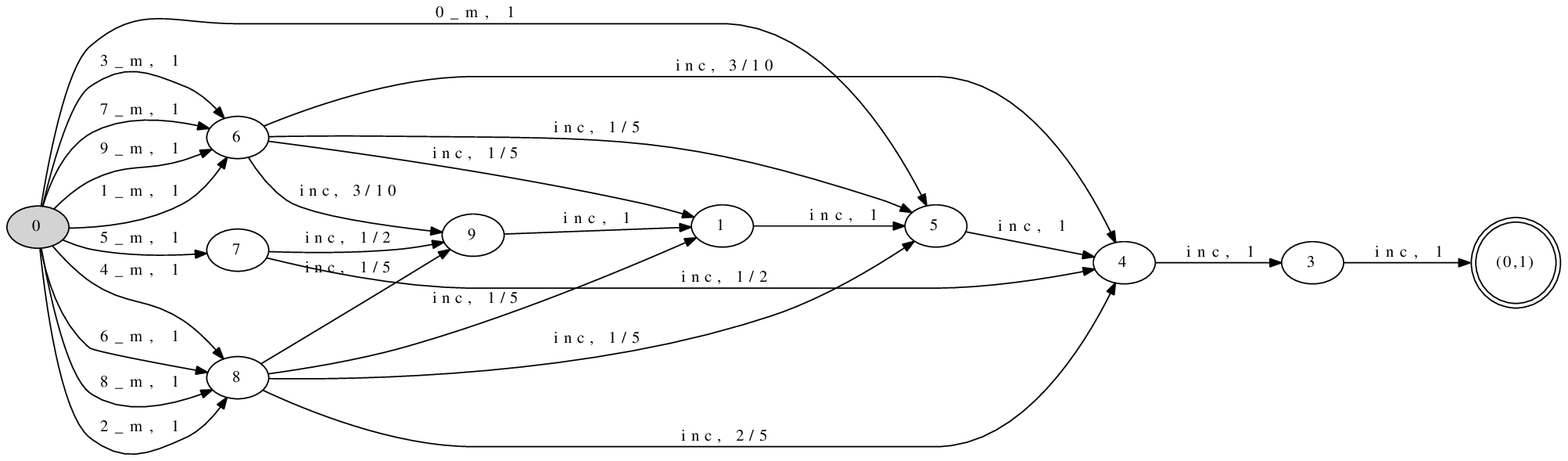}
\end{tabular}
\end{center}
\caption{Modeling RSA decryption with \apex{}.}
\label{fig:ex-rsa}
\end{figure}

\label{ex:rsa}
\end{example}

\section{Pushdown Automata and Arithmetic Circuits}
\label{sec:vpa}

In a visibly pushdown automaton~\cite{AM04} the stack operations are
determined by the input word.  Consequently VPA have a more tractable
language theory than ordinary pushdown automata.  The main result of
this section shows that the equivalence problem for $\Q$-weighted VPA
is logspace equivalent to the problem \textbf{ACIT} of determining
whether a polynomial represented by an arithmetic circuit is
identically zero.

A \emph{visibly pushdown alphabet} $\Sigma = \Sigma_c \cup \Sigma_r \cup
\Sigma_{\mathit{int}}$ consists of a finite set of \emph{calls}
$\Sigma_c$, a finite set of \emph{returns} $\Sigma_r$, and a finite
set of \emph{internal actions} $\Sigma_{\mathit{int}}$.
A visibly pushdown automaton over alphabet $\Sigma$ is
restricted so that it pushes onto the stack when it reads a call,
pops the stack when it reads a return, and leaves the stack untouched
when reading internal actions.  Due to this restriction visibly
pushdown automata only accept words in which calls and returns are
appropriately matched.  Define the set of \emph{well-matched words} to be
$\bigcup_{i \in \mathbb{N}}L_i$, where
$L_0  =  \Sigma_{\mathit{int}}+\{\varepsilon\}$ and
$L_{i+1} = \Sigma_c L_i \Sigma_r + L_i L_i$.

A \emph{$\Q$-weighted visibly pushdown automaton} on alphabet $\Sigma$
is a tuple $\mathcal{A}=(n,\valpha,\veta,\Gamma,M)$, where $n$ is the
number of \emph{states}, $\valpha$ is an $n$-dimensional
\emph{initial} (row) vector, $\veta$ is an $n$-dimensional
\emph{final} (column) vector, $\Gamma$ is a finite \emph{stack
  alphabet}, and $M= (M_c,M_r,M_{\mathit{int}})$ is a tuple of
\emph{matrix-valued transition functions} with types $M_c : \Sigma_c
\times \Gamma \to \Q^{n \times n}$, $M_r : \Sigma_r \times \Gamma \to
\Q^{n \times n}$ and $M_{\mathit{int}} : \Sigma_{\mathit{int}} \to
\Q^{n \times n}$.  If $a \in \Sigma_c$ and $\gamma \in \Gamma$ then
$M_c(a,\gamma)_{i,j}$ gives the weight of an $a$-labelled transition
from state $i$ to state $j$ that pushes $\gamma$ on the stack.  If $a
\in \Sigma_r$ and $\gamma \in \Gamma$ then $M_r(a,\gamma)_{i,j}$ gives
the weight of an $a$-labelled transition from state $i$ to $j$ that
pops $\gamma$ from the stack.

For each well-matched word $u \in \Sigma^*$ we define an $n \times n$
rational matrix $\Ms{\A}(u)$ whose $(i,j)$-th entry denotes the total
weight of all paths from state $i$ to state $j$ along input $u$.  The
definition of $\Ms{\A}(u)$ follows the inductive definition of
well-matched words.  The base cases are $\Ms{\A}(\varepsilon) = I$ and
$\Ms{\A}(a)_{i,j}= M_{\mathit{int}}(a)_{i,j}$.  The inductive cases are
\begin{eqnarray}
\Ms{\A}(uv) & = & \Ms{\A}(u) \cdot \Ms{\A}(v) \label{eq:indcase1} \\ 
\Ms{\A}(aub) & = & \sum_{\gamma \in \Gamma}
M_c(a,\gamma) \cdot \Ms{\A}(u) \cdot M_r(b,\gamma) \, ,
\label{eq:indcase2}
\end{eqnarray}
for $a \in \Sigma_c$, $b \in \Sigma_r$.

The weight assigned by $\A$ to a well-matched word $w$ is defined as
$\A(w):=\valpha \Ms{\A}(u) \veta$.  We say that two $\Q$-weighted
VPA $\A$ and $\B$ are \emph{equivalent} if for each well-matched word
$w$ we have $\A(w)=\B(w)$.

An \emph{arithmetic circuit} is a finite directed acyclic multigraph
whose vertices, called \emph{gates}, have indegree $0$ or $2$.
Vertices of indegree $0$ are called \emph{input gates} and are
labelled with a constant $0$ or $1$, or a variable from the set $\{x_i
: i \in \mathbb{N}\}$.  Vertices of indegree $2$ are called
\emph{internal gates} and are labelled with one of the arithmetic
operations $+$, $*$ or $-$.  We assume that there is a unique gate
with outdegree $0$ called the \emph{output}.  Note that $C$ is a
multigraph, so there can be two edges between a pair of gates, i.e.,
both inputs to a given gate can lead from the same source.  We call a
circuit \emph{variable-free} if all inputs gates are labelled $0$ or
$1$.

The \emph{Arithmetic Circuit Identity Testing (\textbf{ACIT})} problem
asks whether the output of a given circuit is equal to the zero
polynomial.  \textbf{ACIT} is known to be in \textbf{coRP} but it
remains open whether there is a polynomial or even sub-exponential
algorithm for this problem~\cite{ABKM09}.  Utilising the fact
that a variable-free arithmetic circuit of size $O(n)$ can compute
$2^{2^n}$, Allender \emph{et al.}~\cite{ABKM09} give a logspace
reduction of the general \textbf{ACIT} problem to the special case of
variable-free circuits.  Henceforth we assume without loss of
generality that all circuits are variable-free.  Furthermore we recall
that \textbf{ACIT} can be reformulated as the problem of deciding
whether two variable-free circuits using only the arithmetic
operations $+$ and $*$ compute the same number~\cite{ABKM09}.

We have the following proposition:
\begin{proposition} \label{prop:ACIT-to-VPA}
$\mathbf{ACIT}$ is logspace reducible to the equivalence problem for
  $\Q$-weighted visibly pushdown automata.
\end{proposition}
\begin{proof}
Let $C$ and $C'$ be two circuits over basis $\{+,*\}$.  Without loss
of generality we assume that in each circuit the inputs of a
depth-$i$ gate both have depth $i+1$, $+$-nodes have even depth,
$*$-nodes have odd depth, and input nodes all have the same depth $d$.
Notice that in either circuit any path from an input gate to an
output gate has length $d$.

We define two automata $\mathcal{A}$ and $\mathcal{A}'$ that are
equivalent if and only if $C$ and $C'$ have the same output.  Both
automata are defined over the alphabet $\{c,r,\iota\}$, with $c$ a
call, $r$ a return and $\iota$ an internal event.  We explain how
$\mathcal{A}$ arises from $C$; the definition of $\mathcal{A}'$ is
entirely analogous.

Suppose that $C$ has set of gates $\{g_0,g_1,\ldots,g_n\}$, with $g_0$
the output gate.  For each gate $g_i$ of $C$ we include a state $s_i$
of $\mathcal{A}$ and a stack symbol $\gamma_i$.  The initial state of
$\mathcal{A}$ is $s_0$, and all states are accepting.  The transitions
of $\mathcal{A}$ are defined as follows:

\begin{iteMize}{$\bullet$}
\item For each $+$-gate $g_i:=g_j+g_k$ in $C$ we include an internal
  transition from $s_i$ that goes to $s_j$ with probability $1/2$ and
  to $s_k$ with probability $1/2$.
\item For each $*$-gate $g_i := g_j * g_k$ we include a
  probability-$1$ call transition from $s_i$ to $s_j$ that pushes
  $\gamma_k$ onto the stack.
\item
An input gate $g_i$ with label $0$ contributes no transitions.
\item For each input gate $g_i$ with label $1$ and each stack symbol
  $\gamma_j$, we include a return transition from $s_i$ that pops
  $\gamma_j$ off the stack and ends in state $s_j$ with probability
  $1$.
\end{iteMize}
Recall that acceptance is by empty stack and final state.  By
construction $\mathcal{A}$ only accepts a single word, as we now
explain.  Define a sequence of words $w_n \in \{c,r,\iota\}^*$ by $w_0
= \iota$, $w_{n+1} = \iota w_n$ for $n$ even, and
$w_{n+1} =  c w_n r w_n$ for $n$ odd.
Furthermore, write $M_0 = 1$, $M_{n+1} = 2M_n$ for $n$ even, and
$M_{n+1} = M_n^2$ for $n$ odd.  Then $\mathcal{A}$ accepts $w_d$ with
probability $N/M_d$, where $d$ is the depth of the circuit $C$ and $N$
is output of $C$.  All other words are accepted with probability $0$.
We conclude that $C$ and $C'$ have the same value if and only if
$\mathcal{A}$ and $\mathcal{A}'$ are equivalent.
\end{proof}

In the remainder of this section we give a converse reduction: from
equivalence of $\Q$-weighted VPA to \textbf{ACIT}.  The following
result gives a decision procedure for the equivalence of two
$\Q$-weighted VPA $\A$ and $\B$.

\begin{proposition}
$\A$ is equivalent to $\B$ if and only if $\A(w)=\B(w)$ for all
words $w \in L_{n^2}$, where $n$ is the sum of the number of states
of $\A$ and the number of states of $\B$.
\label{prop:equiv}
\end{proposition}
\begin{proof}
Recall that for each balanced word $u
\in \Sigma^*$ we have rational matrices $\Ms{\A}(u)$ and $\Ms{\B}(u)$
giving the respective state-to-state transition weights of $\A$
and $\B$ on reading $u$.  These two families of matrices can be
combined into a single family

\vspace{-4mm}
\[ \mathcal{M} = \left \{ \left(
                 \begin{array}{cc}
                  \Ms{\A}(u) & \mathbf{0}\\
                 \mathbf{0} & \Ms{\B}(u)
                 \end{array}\right)
: \mbox{$u$ well-matched} \right\}\] of $n \times n$ matrices.  Let us
                 also write $\mathcal{M}_i$ for the subset of
                 $\mathcal{M}$ generated by those well-matched words
                 $u \in L_i$.

Let $\alphas{\A},\etas{\A}$ and $\alphas{\B},\etas{\B}$ be the
respective initial and final-state vectors of $\A$ and $\B$.  Then
$\A$ is equivalent to $\B$ if and only if
\begin{gather}
( \begin{array}{cc} \alphas{\A} & \alphas{\B}
\end{array} ) M \left( \begin{array}{c}
                \etas{\A} \\ -\etas{\B}
\end{array} \right) = 0
\label{eq:zero}
\end{gather}
for all $M \in \mathcal{M}$.  It follows that $\A$ is equivalent to
$\B$ if and only if (\ref{eq:zero}) holds for all $M$ in
$\mathrm{span}(\mathcal{M})$, where the span is taken in the rational
vector space of $n \times n$ rational matrices.  But
$\mathrm{span}(\mathcal{M}_i)$ is an ascending sequence of vector spaces:
\[ \mathrm{Span}(\mathcal{M}_0) \subseteq
   \mathrm{Span}(\mathcal{M}_1) \subseteq
   \mathrm{Span}(\mathcal{M}_2) \subseteq \ldots
\]
It follows from a dimension argument that this sequence stops in at
most $n^2$ steps and we conclude that $\mathrm{span}(\mathcal{M})=
\mathrm{span}(\mathcal{M}_{n^2})$.
\end{proof}

\begin{proposition}
  Given a $\Q$-weighted visibly pushdown automaton $\A$ and \mbox{$n
    \in\mathbb{N}$} one can compute in logarithmic space a circuit
  that represents $\sum_{w \in L_{n^2}} \A(w)$.
\label{prop:aut2circuit}
\end{proposition}
\begin{proof}
From the definition of the language $L_i$ and the family of
matrices $\Ms{\A}$ we have:
\begin{eqnarray*}
\sum_{w \in L_{i+1}} \Ms{\A}(w) & = & \sum_{a \in \Sigma_c}
                                    \sum_{b \in \Sigma_r}
                                    \sum_{\gamma \in \Gamma}
 \Ms{\A}(a,\gamma) \left( \sum_{u \in L_i} \Ms{\A}(u) \right) \Ms{\A}(b,\gamma)\\
 & & + \left(\sum_{u \in L_i} \Ms{\A}(u)\right)
       \left(\sum_{u \in L_i} \Ms{\A}(u)\right) \, .
\end{eqnarray*}
The above equation implies that we can compute in logarithmic space a
circuit that represents $\sum_{w \in L_n} \Ms{\A}(w)$.  The result of
the proposition immediately follows by premultiplying by the initial
state vector and postmultiplying by the final state vector.
\end{proof}

A key property of $\Q$-weighted VPA is their closure under product.
\begin{proposition} \label{prop:prod}
Given $\Q$-weighted VPA $\A$ and $\B$ on the same alphabet $\Sigma$ one can
define a \emph{synchronous-product automaton}, denoted $\A \otimes \B$,
such that $(\A \otimes \B)(w) = \A(w) \B(w)$ for all $w \in \Sigma^*$.
\end{proposition}
\begin{proof}
The proof exploits the fact that the stack height is
determined by the input word, so the respective stacks of $\A$ and
$\B$ operating in parallel can be simulated in a single stack.

Let $\A=(\ns{\A},\Sigma,\Gammas{\A},\Ms{\A},\alphas{\A},\etas{\A})$
and $\B=(\ns{\B},\Sigma,\Gammas{\A},\Ms{\B},\alphas{\B},\etas{\B})$.
We define a product automaton $\C$.  Note that since the stack height
is determined by the input word we can simulate the respective stacks
of $\A$ and $\B$ using a single stack in $\C$ whose alphabet is the
product of the respective stack alphabets of $\A$ and $\B$.

The number of states of $\C$ is $\ns{\A} \cdot \ns{\B}$.  The initial
vector $\alphas{\C}$ in the vector $\alphas{\A} \otimes \alphas{\B}$
and the final vector $\etas{\C}$ is $\etas{\A} \otimes \etas{\B}$.
The stack alphabet of $\C$ is $\Gammas{\A} \times \Gammas{\B}$.  Given
$a \in \Sigma_c \cup \Sigma_r$ the transition matrix
$\Ms{\C}(a,(\gamma,\gamma'))$ is $\Ms{\A}(a,\gamma) \otimes
\Ms{\B}(a,\gamma')$.  Likewise, given $a\in \Sigma_{\mathit{int}}$ the
transition matrix $\Ms{\C}(a)$ is $\Ms{\A}(a)\otimes \Ms{\C}(a)$.

It is now straightforward to show that $\Ms{\C}(w)=\Ms{\A}(w)\otimes
\Ms{\B}(w)$ for all balanced words $w\in \Sigma^*$.  The proof
proceeds by induction on balanced words, following (\ref{eq:indcase1})
and (\ref{eq:indcase2}), and using Proposition~\ref{prop:kronecker} on
Kronecker products.
\end{proof}

\begin{proposition}
The equivalence problem for $\Q$-weighted visibly pushdown automata is
logspace reducible to \textbf{ACIT}.
\label{prop:red}
\end{proposition}
\begin{proof}
Let $\A$ and $\B$ be $\Q$-weighted visibly pushdown automata with a
total of $n$ states between them.  Then
\begin{eqnarray*}
 \sum_{w \in L_n} (\A(w)-\B(w))^2 & = &
 \sum_{w \in L_n} \A(w)^2 + \B(w)^2 - 2\A(w)\B(w) \\
& = &
\sum_{w \in L_n} (\A \otimes \A)(w) + (\B \otimes \B)(w) - 2(\A \otimes \B)(w)
\end{eqnarray*}
Thus $\A$ is equivalent to $\B$ iff \mbox{$\sum_{w \in L_n} (\A \otimes
  \A)(w)+ (\B \otimes \B)(w) = 2\sum_{w \in L_n} (\A \otimes \B)(w)$}.
But Propositions \ref{prop:aut2circuit} and~\ref{prop:prod} allow us
to translate the above equation into an instance of \textbf{ACIT}.
\end{proof}

The trick of considering sums-of-squares of acceptance weights in the
above proof is inspired by~\cite[Lemma 1]{Tzeng96}.

\section{Conclusion}
It is known that deciding equivalence of $\Q$-weighted finite automata is in
\textbf{NC}~\cite{Tzeng96}.  We have shown that deciding minimality is
also in \textbf{NC}.  Regarding the corresponding function problems,
we have given an \textbf{RNC} algorithm to decide equivalence and
output a counterexample word in case the input automata differ, and an
\textbf{RNC} algorithm to minimise an automaton.  We do not know
whether either of these problems is in \textbf{NC}.  It would be
interesting to explore whether there is a relationship between these
two problems, and to relate them to other problems in \textbf{RNC}
that are not known to be in \textbf{NC}, such as bipartite matching.

For $\Q$-weighted VPA the situation is more complete.  We have shown
that deciding equivalence is equivalent to polynomial identity
testing, the complexity of which is an important open problem.

\bibliographystyle{plain} 
\bibliography{db}

\end{document}